\documentclass[a4paper,UKenglish]{lipics-v2018}

\usepackage{microtype}
\usepackage{algorithm}
\nolinenumbers

\usepackage[]{algorithmicx}
\usepackage{amsfonts,amstext,amsmath,amssymb,stmaryrd,mathrsfs,calc,amsthm,xspace,mathtools}

\usepackage{bm}
\usepackage{comment}
\usepackage{framed}
\usepackage{thmtools}
\usepackage{thm-restate}
\usepackage{graphics}
\usepackage{enumerate}
\usepackage[ color=white, textsize=small]{todonotes}

\usepackage{mathtools}

\newcommand{\?}[1]{\mathcal{#1}}

\let\originalleft\left
\let\originalright\right
\renewcommand{\left}{\mathopen{}\mathclose\bgroup\originalleft}
\renewcommand{\right}{\aftergroup\egroup\originalright}

\newcommand{\SAT}[1]{\mathit{SAT}(#1)}
\newcommand{\Var}{\mathit{Var}}

\newcommand{\cmax}{c_\mathit{max}}
\newcommand{\Pre}{\mathit{Pre}}
\newcommand{\Post}{\mathit{Post}}

\newcommand{\MS}[1]{{#1}^\oplus}

\newcommand{\Lang}[1]{\?{L}\left(#1\right)}
\newcommand{\abstr}[2][]{\mathit{abs}_{#1}(#2)}

\newcommand{\NDStep}[1]{\xrightarrow{#1}_{\textit{Disc}}}
\newcommand{\NTStep}[1]{\xrightarrow{#1}_{\textit{Time}}}
\newcommand{\tstep}[1]{\longrightarrow_{#1}}
\newcommand{\dstep}[1]{\longrightarrow_{#1}}

\newcommand{\step}[1]{\xrightarrow{#1}}
\newcommand{\Tstep}{\NTStep{}}
\newcommand{\Dstep}[1]{\NDStep{#1}}

\newcommand{\Coverset}[2][]{\mathit{Cover}_{#1}\left(#2\right)}
\newcommand{\Reachset}[2][]{\mathit{Cover}_{#1}(#2)}
\newcommand{\x}{\times}
\newcommand{\N}{\mathbb{N}}

\newcommand{\nat}{{\mathbb N}}
\newcommand{\nnreals}{{\mathbb R}_{\geq 0}}

\renewcommand{\vec}[1]{\boldsymbol{#1}}

\newcommand{\abs}[1]{\lvert#1\rvert}
\newcommand{\card}[1]{\abs{#1}}
\newcommand{\len}[1]{\abs{#1}}

\newcommand{\eqby}[2][=]{\overset{\text{\tiny#2}}{#1}}
\newcommand{\eqdef}{\eqby{def}}

\newcommand{\PSPACE}{\ensuremath{\mathsf{PSPACE}}}
\newcommand{\EXPSPACE}{\ensuremath{\mathsf{EXPSPACE}}}
\newcommand{\dproblem}[2]{
\medskip
\noindent
\fbox{
\parbox{\textwidth-1em}{\textbf{Input}: {#1}
\\[+2pt]
\textbf{Question}: {#2 }
}
}
}

\newcommand{\ECOVER}{$\exists$COVER }

\newcommand{\hide}[1]{}

\newcommand{\set}[1]{\left\{#1\right\}}

\newcommand{\denotationof}[1]{[\![#1]\!]}

\newcommand{\zeroto}[1]{[{#1}]}

\newcommand{\fract}{{\it frac}}
\newcommand{\fractof}[1]{{\it frac}(#1)}
\newcommand{\intof}[1]{{\it int}(#1)}

\newcommand{\down}[1]{\downarrow\!{#1}}
\newcommand{\aset}{{\it A}}

\newcommand{\eps}{\varepsilon}

\newcommand{\wordsover}[1]{{#1}^*}

\usepackage[normalem]{ulem}
\usepackage[capitalise,noabbrev,nameinlink]{cleveref}
\usepackage{thm-restate}

\usepackage{algorithm}
\usepackage[noend]{algpseudocode}

\makeatletter
\let\c@theorem\@undefined
\let\theorem\@undefined
\let\endtheorem\@undefined
\let\lemma\@undefined
\let\endlemma\@undefined
\let\corollary\@undefined
\let\endcorollary\@undefined
\let\definition\@undefined
\let\enddefinition\@undefined
\let\example\@undefined
\let\endexample\@undefined
\let\remark\@undefined
\let\endremark\@undefined
\makeatother
\theoremstyle{plain}
\newtheorem{theorem}{Theorem}
\newtheorem{lemma}[theorem]{Lemma}
\newtheorem{corollary}[theorem]{Corollary}
\theoremstyle{definition}
\newtheorem{definition}[theorem]{Definition}
\newtheorem{example}[theorem]{Example}
\theoremstyle{remark}
\newtheorem{remark}[theorem]{Remark}
\crefname{lemma}{Lemma}{Lemmas}
\crefname{definition}{Definition}{Definitions}
\theoremstyle{plain}

\crefname{proposition}{Proposition}{Propositions}
\theoremstyle{remark}
{\itshape}{\rmfamily}
\crefname{fact}{Fact}{Facts}
\newtheorem{claim}{Claim}[theorem]
\crefname{claim}{Claim}{Claims}
\usepackage{tikz}
\usetikzlibrary{positioning}
\usetikzlibrary{arrows}
\usetikzlibrary{petri}
\tikzset{>=latex}
\tikzstyle{every place}=[circle,draw=blue!50,fill=blue!20,thick,inner sep=2mm,]
\tikzstyle{every transition}=[rectangle,draw=black!50,fill=black!20,thick,inner sep=2mm,rounded corners=3pt]

\usepackage{ifthen}
\usepackage{chngcntr}
\usepackage{apptools}
\AtAppendix{\counterwithin{lemma}{section}}

\title{Universal Safety for Timed Petri Nets is PSPACE-complete}

\author{Parosh Aziz Abdulla}{Uppsala University, Sweden}{}{}{}
\author{Mohamed Faouzi Atig}{Uppsala University, Sweden}{}{}{}
\author{Radu Ciobanu}{University of Edinburgh, UK}{}{}{}
\author{Richard Mayr}{University of Edinburgh, UK}{}{}{}
\author{Patrick Totzke}{University of Edinburgh, UK}{}{https://orcid.org/0000-0001-5274-8190}{}

\authorrunning{P.A.~Abdulla, M.~Atig, R.~Ciobanu, R.~Mayr and P.~Totzke}

\Copyright{Parosh Aziz Abdulla, Mohamed Faouzi Atig, Radu Ciobanu, Richard Mayr and Patrick Totzke}

\subjclass{
}
\category{}

\relatedversion{}

\supplement{}

\funding{This work was supported by the EPSRC, grant EP/M027651/1.}

\acknowledgements{}


\EventEditors{Sven Schewe and Lijun Zhang}
\EventNoEds{2}
\EventLongTitle{29th International Conference on Concurrency Theory (CONCUR 2018)}
\EventShortTitle{CONCUR 2018}
\EventAcronym{CONCUR}
\EventYear{2018}
\EventDate{September 4--7, 2018}
\EventLocation{Beijing, China}
\EventLogo{}
\SeriesVolume{118}
\ArticleNo{6} 
\usepackage{tikz}
\usetikzlibrary{positioning}
\usetikzlibrary{arrows}
\usetikzlibrary{calc}

\tikzset{>=latex}
\tikzstyle{graphnode} = [circle,
thick,
inner sep=0mm,
minimum size=2mm,
text width=2mm,
align=center]

\tikzstyle{newnode} = [graphnode, draw=red!50,fill=red!25,]
\tikzstyle{snode} = [graphnode, draw=blue!50,fill=blue!25]
\tikzstyle{mnode} = [graphnode, draw=green!50,fill=green!25, ]

\begin{document}
\maketitle
\begin{abstract}
A timed network consists of an arbitrary number of initially identical 1-clock
timed automata, interacting via hand-shake communication.
In this setting there is no unique central controller, since all automata are 
initially identical. 
We consider the universal safety problem for such controller-less timed networks, 
i.e., verifying that a bad event (enabling some given transition) 
is impossible regardless of the size of the network. 

This universal safety problem is dual to the 
existential coverability problem for timed-arc Petri nets, i.e.,
does there exist a number $m$ of tokens, 
such that starting with $m$ tokens in a given place, and none in the other
places, some given transition is eventually enabled.

We show that these problems are PSPACE-complete.

\end{abstract}

\section{Introduction}
\label{sec:intro}
\subparagraph{Background.}
Timed-arc Petri nets (TPN) \cite{AA2001,S2005,AMM07,BHR2006,JJMS2011} are an
extension of Petri nets where each token carries one real-valued clock and
transitions are guarded by inequality constraints where the clock values are
compared to integer bounds (via strict or non-strict inequalities).
The known models differ slightly in what clock values newly created tokens can
have, i.e., whether newly created tokens can inherit the clock value of some
input token of the transition, or whether newly created tokens always have
clock value zero. We consider the former, more general, case.

Decision problems associated with the reachability analysis of (extended)
Petri nets include 
\emph{Reachability} (can a given marking reach another given marking?) 
and \emph{Coverability} (can a given marking ultimately enable a given transition?).

While Reachability is undecidable for all these TPN models \cite{RVCE1999},
Coverability is decidable using the well-quasi ordering approach of \cite{ACJT2000,FS2001}
and complete for the hyper-Ackermannian complexity class $F_{\omega^{\omega^\omega}}$ \cite{HSS2012}.
With respect to Coverability, TPN are equivalent \cite{BFHR2010} to (linearly ordered) data nets \cite{LNORW2008}.

The \emph{Existential Coverability} problem for TPN asks, for a given place
$p$ and transition $t$, whether there exists a number $m$ 
such that the marking $M(m) \eqdef m\cdot\{(p,\vec{0})\}$ 
ultimately enables $t$. Here, $M(m)$ contains exactly $m$ tokens on place $p$ with all clocks set to zero
and \emph{no other tokens}.
This problem corresponds to checking safety properties in distributed networks of
arbitrarily many (namely $m$) initially identical timed processes that communicate by handshake.
A negative answer certifies that the `bad event'
of transition $t$ can never happen regardless of the number $m$ of processes,
i.e., the network is safe for any size.
Thus by checking existential coverability, one solves the dual problem
of \emph{Universal Safety}. 
(Note that the $m$ timed tokens/processes are only initially identical. They
can develop differently due to non-determinacy in the transitions.)

The corresponding problem for timed networks studied
in \cite{ADM2004} does not allow the
dynamic creation of new timed processes (unlike the TPN model which can increase the
number of timed tokens), but considers multiple clocks per process (unlike our
TPN with one clock per token).

The TPN model above corresponds to a distributed network without a central
controller, since initially there are no tokens on other places that could be
used to simulate one.
Adding a central controller would make \emph{Existential Coverability}
polynomially inter-reducible with normal \emph{Coverability} and thus
complete for $F_{\omega^{\omega^\omega}}$ \cite{HSS2012} (and even undecidable for $>1$
clocks per token \cite{ADM2004}).

Aminof et.~al.~\cite{ARZS2015} study the
model checking problem of $\omega$-regular properties for the controller-less 
model and in particular claim an \EXPSPACE\ upper bound for checking universal
safety. However, their result only holds for discrete time (integer-valued clocks) 
and they do not provide a matching lower bound.

\subparagraph{Our contribution.}
We show that \emph{Existential Coverability} (and thus universal safety)
is decidable and \PSPACE-complete.
This positively resolves an open question from \cite{ADM2004} regarding the decidability of universal safety in the controller-less networks.
Moreover, a symbolic representation of the set of coverable configurations can
be computed (using exponential space).

The \PSPACE\ lower bound is shown by a reduction from the iterated monotone
Boolean circuit problem. (It does not follow directly 
from the \PSPACE-completeness of the reachability problem in
timed automata of \cite{AD:1994}, due to the lack of a central controller.)

The main ideas for the \PSPACE\ upper bound are as follows.
First we provide a logspace reduction of the Existential Coverability problem
for TPN to the corresponding problem for a syntactic subclass, non-consuming
TPN.
Then we perform an abstraction of the real-valued clocks, similar
to the one used in \cite{AMM07}. Clock values are split into integer parts 
and fractional parts. The integer parts of the clocks can be abstracted into a
finite domain, since the transition guards cannot distinguish between values
above the maximal constant that appears in the system. 
The fractional parts of the clock values that occur in a marking
are ordered sequentially.
Then every marking can be abstracted into a string where all the tokens with 
the $i$-th fractional clock value are encoded in the $i$-th symbol in the
string. Since token multiplicities do not matter for existential coverability,
the alphabet from which these strings are built is finite.
The primary difficulty is that the length of these strings 
can grow dynamically as the system evolves, i.e., the space of
these strings is still infinite for a given TPN.
We perform a forward exploration of the space of reachable strings.
By using an acceleration technique, we can effectively construct a symbolic
representation of the set of reachable strings in terms
of finitely many regular expressions.
Finally, we can check existential coverability by using this symbolic
representation.

\section{Timed Petri Nets}
We use $\nat$ and $\nnreals$ to denote the sets
of nonnegative integers and reals, respectively. 
For $n\in\nat$ we write $\zeroto{n}$ for the set $\set{0,\ldots,n}$.

For a set $\aset$, we use $\wordsover\aset$ to denote the set of words, i.e.~finite sequences,
over $\aset$, and write $\eps$ for the empty word.
If $R$ is a regular expression over $\aset$ then $\Lang{R}\subseteq \aset^*$ denotes its language.

A \emph{multiset} over a set $X$ is a function $M:X\to\N$.
The set $\MS{X}$ of all (finitely supported) multisets over $X$ is partially ordered pointwise (by $\le$).
The multiset union of $M,M'\in\MS{X}$
is $(M\oplus M')\in\MS{X}$ with $(M\oplus M')(\alpha)\eqdef M(\alpha)+M'(\alpha)$ for all $\alpha\in X$.
If $M\ge M'$ then the multiset difference $(M\ominus M')$ is the unique $M''\in\MS{X}$
with $M=M'\oplus M''$. We will use a monomial representation and write for example $(\alpha + \beta^3)$
for the multiset $(\alpha\mapsto1, \beta\mapsto 3)$.
For a multiset $M$ and a number $m\in\N$ we let
$m\cdot M$ denote the $m$-fold multiset sum of $M$.
We further lift this to sets of numbers and multisets
on the obvious fashion, so that in particular
$\N\cdot S \eqdef \{n\cdot M\mid n\in \N, M\in S\}$.

\label{sec:model}
\medskip
\emph{Timed Petri nets} are place/transition nets where each token carries a real value,
sometimes called its \emph{clock value} or \emph{age}.
Transition firing depends on there being sufficiently many tokens whose value is in a specified interval.
All tokens produced by a transition either have age $0$, or inherit the age of
an input-token of the transition.
To model time passing, all token ages can advance simultaneously by the same (real-valued) amount.

\begin{definition}[TPN]
    \label{def:MTPN}
A \emph{timed Petri net} (TPN) 
$\?N = (P,T, \Var, G, \Pre,\Post)$ 
consists of finite sets of \emph{places} $P$, \emph{transitions} $T$
and \emph{variables} $\Var$, as well as functions
$G,\Pre,\Post$ defining
transition \emph{guards}, \emph{pre}-- and \emph{postconditions}, as follows.

For every transition $t\in T$, the guard $G(t)$ 
maps variables to (open, half-open or closed) 
intervals with endpoints in $\N\cup\{\infty\}$,
restricting which values variables may take.
All numbers are encoded in unary.
The precondition $\Pre(t)$ is a finite multiset over $(P\x\Var)$.
Let $\Var(t)\subseteq\Var$ be the subset of variables appearing positively in $\Pre(t)$.
The postcondition $\Post(t)$ is then a finite multiset over
$(P\x(\{0\}\cup\Var(t)))$,
specifying the locations and clock values of produced tokens.
Here, the symbolic clock value is either
$0$ (demanding a reset to age $0$), or a variable that appeared already in the precondition.

A \emph{marking} is a finite multiset over $P\x\nnreals$.
\end{definition}

\begin{example}
\label{ex-mtpn}
The picture below shows a place/transition representation of an TPN with four places and one transition.
$\Var(t)=\{x,y\}$, $\Pre(t) = (p,x)^2 + (q,y)$, $G(t)(x)=[0,5]$, $G(t)(y)=]1,2]$
and $\Post(t)= (r,y)^3 + (s,0)$.

\begin{center}
\begin{tikzpicture}[node distance=0.5cm and 3.5cm, on grid,->]
    \node (t) [transition, label=above:{$t$},align=left]
        {$0\le x \le 5$\\
        $1 < y \le 2$
        };
\node (p) [above left=of t,place, label=left:{$p$}] {};
\node (q) [below left=of t,place, label=left:{$q$}] {};
\node (r) [above right=of t,place, label=right:{$r$}] {};
\node (s) [below right=of t,place, label=right:{$s$}] {};

\draw (p) edge node[above right, pos=0.3]{$x^2$} (t);
\draw (q) edge node[below right, pos=0.3]{$y$} (t);
\draw (t) edge node[above]{$y^3$} (r);
\draw (t) edge node[above]{$0$} (s);
\end{tikzpicture}
\end{center}
The transition $t$ consumes two tokens from place $p$, both of which have the
same clock value $x$ (where $0 \le x \le 5$) and one token from place $q$ with
clock value $y$ (where $1 < y \le 2$).
It produces three tokens on place $r$ who all have the same clock value $y$
(where $y$ comes from the clock value of the token read from $q$),
and another token with value $0$ on place $s$.
\end{example}

There are two different binary step relations on markings:
\emph{discrete} steps $\dstep{t}$ which fire a transition $t$ as specified by the
relations $G,\Pre$, and $\Post$,
and \emph{time passing} steps $\tstep{d}$ for durations $d\in\nnreals$, which simply increment all clocks by $d$.

\begin{definition}[Discrete Steps]
    \label{def:mtpn:dsteps}
    For a transition $t\in T$ and a variable evaluation $\pi:\Var\to \nnreals$,
    we say that \emph{$\pi$ satisfies
    $G(t)$} if $\pi(x)\in G(t)(x)$ holds for all $x\in \Var$.
    By lifting $\pi$ to multisets over $(P\x\Var)$ 
    (respectively, to multisets over $(P\x(\{0\}\cup\Var))$ with $\pi(0)=0$)
    in the canonical way,
    such an evaluation translates preconditions 
    $\Pre(t)$ and $\Post(t)$
    into markings
    $\pi(\Pre(t))$ and $\pi(\Post(t))$,
    where for all $p\in P$ and $c\in\nnreals$,
    \begin{align*}
        \pi(\Pre(t))(p,c)\eqdef \sum_{\pi(v)=c} \Pre(t)(p,v)
        \qquad
        \text{and}
        \qquad
        \pi(\Post(t))(p,c)\eqdef \sum_{\pi(v)={c}} \Post(t)(p,{v}).
    \end{align*}
    A transition $t\in T$ is called \emph{enabled} in marking $M$, if
    there exists an evaluation $\pi$ that satisfies $G(t)$ and such that
    $\pi(\Pre(t))\le M$.
    In this case, there is a discrete step $M\dstep{t}M'$ from marking $M$ to $M'$, defined as
        $
        M' = M \ominus \pi(\Pre(t)) \oplus \pi(\Post(t)).
        $
\end{definition}

\begin{definition}[Time Steps]
    \label{def:mtpn:tsteps}
    Let $M$ be a marking and $d\in \nnreals$.
    There is a time step $M\tstep{d}M'$
    to the marking $M'$ 
    with
    $M'(p,{c})\eqdef M(p,{c}-{d})$ for ${c}\ge{d}$,
    and
    $M'(p,{c})\eqdef 0$, otherwise.
    We also refer to $M'$ as $(M+d)$.
\end{definition}

We write $\NTStep{}$ for the union of all timed steps,
$\Dstep{}$ for the union of all discrete steps and simply $\step{}$ for $\Dstep{}\cup\Tstep{}$.
The transitive and reflexive closure of $\step{}$ is $\step{*}$.
$\Coverset{M}$ denotes the set of markings $M'$
for which there is an $M''\ge M'$ with $M\step{*}M''$.

We are interested in the \emph{existential coverability problem}
(\ECOVER\ for short), as follows.

\medskip
\dproblem{
    A TPN, an initial place $p$ and a transition $t$.
}
{
    Does there exist $M\in \Coverset{\N\cdot \{(p,{0})\}}$
    that enables $t$?
}

\medskip
\noindent
We show that this problem is \PSPACE-complete.
Both lower and upper bound will be shown (w.l.o.g., see \cref{lem:mtpn:wlog}) for the syntactic subclass of \emph{non-consuming} TPN, defined as follows.

\begin{definition}
\label{def:nonconsuming}
A \emph{timed Petri net} $(P,T, \Var, G, \Pre,\Post)$ is \emph{non-consuming} if 
for all $t\in T$, $p\in P$ and $x\in\Var$
it holds that both
1) $\Pre(t)(p,x) \le 1$, and
2) $\Pre(t) \le \Post(t)$.
\end{definition}
In a non-consuming TPN, token multiplicities are irrelevant for discrete transitions. 
Intuitively, having one token $(p,{c})$ is equivalent 
to having an inexhaustible supply 
of 
such tokens. 

The first condition is merely syntactic convenience.
It asks that each transition takes at most one token from each place. 
The second condition in \cref{def:nonconsuming}
implies that for each discrete step $M\dstep{t}M'$ we have $M' \ge M$.
Therefore,
once a token $(p,{c})$ is present on a place $p$, it will stay there unchanged (unless
time passes), and it will enable transitions
with $(p,{c})$ in their precondition.

Wherever possible, we will from now on therefore allow ourselves to use the set notation
for markings, that is simply treat markings
$M\in\MS{(P\x\nnreals)}$ as sets $M\subseteq (P\x\nnreals)$.

\section{Lower Bound}
\label{sec:lowerbound}
\PSPACE-hardness of \ECOVER does not follow directly 
from the \PSPACE-completeness of the reachability problem in
timed automata of \cite{AD:1994}.
The non-consuming property of our TPN makes it impossible to 
fully implement the control-state of a timed automaton.
Instead our proof uses multiple timed tokens and a reduction from
the iterated monotone Boolean circuit problem \cite{GMST2016}.

    \medskip
    A depth-1 monotone Boolean circuit is a function $F:\{0,1\}^n \to \{0,1\}^n$
    represented by $n$ constraints: For every $0\le i < n$ there is a constraint of the form
    $i'=j \otimes k,$
    where $0\le j,k< n$ and $\otimes \in \{\wedge, \vee \}$,
    which expresses how the next value of bit $i$ depends on the current values of bits $j$ and $k$.
    For every bitvector $\vec{v}\in\{0,1\}^n$, the function $F$ then satisfies $F(\vec{v})[i]\eqdef \vec{v}[j]\otimes \vec{v}[k]$.
    It is \PSPACE-complete to check whether for a given vector $\vec{v}\in\{0,1\}^n$ 
    there exists a number $m\in\N$ such that $F^m(\vec{v})[0]=1$.

    \newcommand{\True}[1]{\mathit{True}_{#1}}
    \newcommand{\False}[1]{\mathit{False}_{#1}}
    \medskip
    Towards a lower bound for \ECOVER~(\cref{thm:LB})
    we construct a non-consuming TPN as follows, for a given circuit.
    The main idea is to simulate circuit constraints by transitions that
    reset tokens of age $1$ (encoding $\vec{v}$) to fresh ones of age $0$ (encoding $F(\vec{v})$),
    and let time pass by one unit to enter the next round.

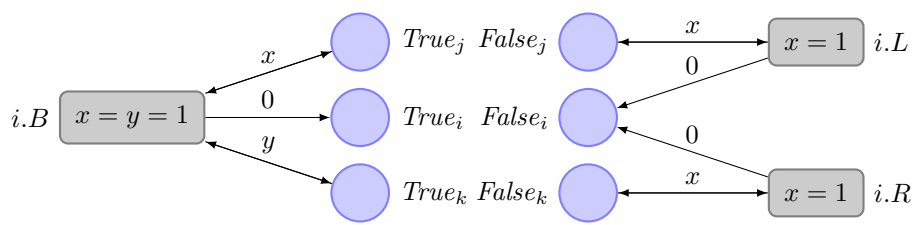
\begin{figure}[t]
\centering
\begin{tikzpicture}[node distance=1cm and 3cm, on grid,->]
    \node (Tj) [place, label=right:{$\mathit{True}_j$}] {};
    \node (Ti) [place, below=of Tj, label=right:{$\mathit{True}_{i}$}] {};
    \node (Tk) [place, below=of Ti, label=right:{$\mathit{True}_{k}$}] {};
    \node (Fj) [place, right=of Tj, label=left:{$\mathit{False}_j$}] {};
    \node (Fi) [place, right=of Ti, label=left:{$\mathit{False}_{i}$}] {};
    \node (Fk) [place, right=of Tk, label=left:{$\mathit{False}_{k}$}] {};
    
    \node (iB) [transition, left=of Ti, label=left:{$i.B$}] {$x=y=1$};
    \node (iL) [transition, right=of Fj, label=right:{$i.L$}] {$x=1$};
    \node (iR) [transition, right=of Fk, label=right:{$i.R$}] {$x=1$};
\draw (Tj) edge node[above]{$x$} (iB);
\draw (Tk) edge node[above]{$y$} (iB);
\draw (iB) edge node[above]{$0$} (Ti);

\draw (Fj) edge node[above]{$x$} (iL);
\draw (iL) edge node[above]{$0$} (Fi);

\draw (iL) edge node[above]{} (Fj);
\draw (iR) edge node[above]{ }(Fk);
\draw (iB) edge node[above]{}(Tj);
\draw (iB) edge node[above]{}(Tk);

\draw (Fk) edge node[above]{$x$} (iR);
\draw (iR) edge node[above]{$0$} (Fi);
\end{tikzpicture}
\caption{The transitions $i.B, i.R$ and $i.L$ that simulate the update of bit $i$
    according to constraint $i'=j\land k$.
    All transitions demand that incoming tokens are of age exactly $1$ and only tokens of age $0$ are produced.
\label{fig:LB}}
\end{figure}

    \bigskip
    \noindent
    For every bit $0\le i<n$, the net contains two places $\True{i}$ and $\False{i}$.
A marking $M_{\vec{v}} \le P\x\nnreals$ is an \emph{encoding} of a vector  $\vec{v}\in\{0,1\}^n$ if for every $0 \leq i < n$ the following hold.
   \begin{enumerate}
    \item $(\True{i},0) \in M_{\vec{v}}\iff \vec{v}[i]=1$.
     \item $(\False{i},0) \in M_{\vec{v}}\iff \vec{v}[i]=0$.
     \item If $(p,c)\in M_{\vec{v}} $ then  $c=0$ or $c\geq 1$.
   \end{enumerate}

   Note that in particular one cannot have both $(\True{i},0)$ and $(\False{i},0)$ in $M_{\vec{v}}$.
    For every constraint $i'=j\land k$
    we introduce three transitions, $i.L, i.R$, and $i.B$, where
    \begin{align*}
        \Pre(i.B) &\eqdef {(\True{j},x) + (\True{k},y)}
        &&
        \Post(i.B) \eqdef \Pre(i.B) + {(\True{i},0)}\\
        \Pre(i.L) &\eqdef {(\False{j},x)}
        &&
        \Post(i.L) \eqdef \Pre(i.L) + (\False{i},0)\\
        \Pre(i.R) &\eqdef {(\False{k},x)}
        &&
        \Post(i.R) \eqdef \Pre(i.R) + (\False{i},0)
    \end{align*}
    and the guard for all transitions is $G(x)=G(y)=1$. See \cref{fig:LB} for an illustration.
    For disjunctions $i'=j\lor k$ the transitions are defined analogously, with $\True{}$ and $\False{}$ inverted.
    The correctness proof of our construction rests on the following simple observation.

    \begin{restatable}{lemma}{claimone}\label{claim1}
       If $F(\vec{v})=\vec{v}'$ then for every
       encoding $M_{\vec{v}}$ of $\vec{v}$, there exists an encoding  $M_{\vec{v'}}$ of $\vec{v}'$ such that $M_{\vec{v}}\tstep{1}\Dstep{*}M_{\vec{v}'}$.
       Conversely, if $M_{\vec{v}}\tstep{1}\Dstep{*}M_{\vec{v}'}$ for encodings $M_{\vec{v}}$ and $M_{\vec{v'}}$ of $\vec{v}$ and  $\vec{v}'$ respectively, then  $F(\vec{v}) = \vec{v'}$.
     \end{restatable}
\begin{proof}
    For the first part,
    we construct a sequence $M_0 \Dstep{} M_1 \Dstep{} \dots \NDStep{} M_{n-1} $ where $M_0 \eqdef (M_{\vec{v}}+1)$ and every step $M_{i-1} \Dstep{} M_i$ adds tokens simulating the $i$th constraint
 of $F$. Since the TPN is non-consuming, we will have that $M_i\ge (M_{\vec{v}}+1)$, for all $i<n$.
Consider now constraint $i'$, and assume w.l.o.g.~that $i' = j \land k$ (the other case is analogous). There are two cases depending on $\vec{v'}[i]$.
 
  \begin{enumerate}
      \item Case $\vec{v'}[i]=1$. By our assumption that $F(\vec{v})=\vec{v'}$ we know that $\vec{v}[j] =1$ and $\vec{v}[k] =1$. So $(\True{j},1) \in (M_{\vec{v}}+1) \le M_{i-1}$ and 
    $(\True{k},1) \in (M_{\vec{v}}+1) \le M_{i-1}$.
    By construction of the net, 
    there is a transition $i.B$ with $\Pre(i.B) = {(\True{j},1) + (\True{k},1)}$ and $\Post(i.B) = \Pre(i.B) + {(\True{i},0)} $.
    This justifies step
    $M_{i-1} \dstep{i.B} M_{i}$ and therefore that $(True_i, 0) \in M_i\le M_{n-1}$.
    Also notice that no marking reachable from $M_0$ using only discrete steps can
    contain the token $(\False{i},0)$.
    This is because these can only be produced by transitions requiring either
    $(\False{j},1)$ or $(\False{k},1)$, which are not contained in $M_0$ by assumption that $M_{\vec{v}}$ encodes $\vec{v}$.
    Therefore $(\False{i},0)\notin M_{n-1}$.
    
\item Case $\vec{v'}[i]=0$.  W.l.o.g., $\vec{v}[j] =0$. 
    Therefore, $(\False{j},1) \in (M_{\vec{v}}+1) \le M_{i-1}$.  
    By construction of the net, there exists transition $i.L$ with $\Pre(i.L) = {(\False{j},1)}$ and $\Post(i.L) = \Pre(i.L) + {(\False{i},0)}$. 
    This justifies the step $M_{i-1} \dstep{i.L} M_{i}$, with $(False_i, 0) \in M_i\le M_{n-1}$.
    Notice again that no marking reachable from $M_0$ using only discrete steps can
    contain the token $(\True{i},0)$.
    This is because these can only be produced by transitions $i.B$,
    requiring both $(\True{j},1),(\True{k},1)\in M_{0}$, contradicting our
    assumptions. Hence, $(\True{i},0)\notin M_{n-1}$.
  \end{enumerate}
We conclude that the constructed marking $M_{n-1}$ is an encoding of $\vec{v'}$.
 
\smallskip
For the other part of the claim,
assume that  there exist markings $M_{\vec{v}}$ and $M_{\vec{v'}}$ which are encodings of vectors $\vec{v}$ and  $\vec{v'}$, respectively, with
$M_{\vec{v}}\tstep{1}\Dstep{*}M_{\vec{v'}}$.
We will show that   $F(\vec{v})=\vec{v}'$. 
Recall that
$F(\vec{v})[i]\eqdef \vec{v}[j]\otimes \vec{v}[k]$, where $0\le j,k< n$ and $\otimes \in \{\wedge, \vee \}$.
We will show for each $i <  n$ that
$\vec{v}'[i] = \vec{v}[j]\otimes \vec{v}[k]$.
Again, consider the constraint $i'$, and assume w.l.o.g.~that $i' = j \land k$ (the other case is analogous). There are two cases.
\begin{enumerate}
    \item Case $\vec{v'}[i]=1$.
        By definition of a marking encoding, we have that $(\True{i},0)\in M_{\vec{v}}$.
By construction, there is a transition $i.B$ with $\Pre(i.B) = {(\True{j},1) + (\True{k},1)}$ and $\Post(i.B) = \Pre(i.B) + {(\True{i},0)}$. 
By assumption, it holds that  $(M_{\vec{v}}+1) \Dstep{*} M_{\vec{v}}'$, where $M_{\vec{v}}\tstep{1} (M_{\vec{v}}+1)$.
Note that $(\True{j},1)\in (M_{\vec{v}}+1)$ and $(\True{k},1)\in (M_{\vec{v}}+1)$.
Hence, we have that $\vec{v}[j]=1$ and $\vec{v}[k]=1$, and therefore that $F(\vec{v})[i] = \vec{v'}[i] = \vec{v}[j] \land \vec{v}[k]$.

\item Case $\vec{v'}[i]=0$. 
Then $(\False{i},0)\in M_{\vec{v}}$
and, since this token can only be produced by transitions $i.L$ or $i.R$,
either $(\False{j},1)\in (M_{\vec{v}}+1)$ or $(\False{k},1)\in (M_{\vec{v}}+1)$.

Therefore 
$(\False{j},0)\in (M_{\vec{v}})$ or $(\False{k},0)\in (M_{\vec{v}})$
and because $M_{\vec{v}}$ is an encoding of $\vec{v}$,
this means that either $\vec{v}[j]=0$ or $\vec{v}[k]=0$.
Therefore, $F(\vec{v'})[i]= \vec{v}[j]\land \vec{v}[k] =0$.  
\qedhere
\end{enumerate}
\end{proof}

\begin{theorem}
    \label{thm:LB}
     \ECOVER
    is \PSPACE-hard for non-consuming TPN.
\end{theorem}
\begin{proof}
    For a given monotone Boolean circuit, define a non-consuming TPN as above.
    By induction on $m\in\N$
    using \cref{claim1}, we derive that
    there exists $m\in\N$ with $F^m(\vec{v})=\vec{v}'$ and $\vec{v}'[0]=1$ if, and only if,
    there exists encodings $M_{\vec{v}}$ 
    of $\vec{v}$ and 
    $M_{\vec{v'}}$ of $\vec{v'}$,
    with $M_{\vec{v}}\step{*} M_{\vec{v}'}$.
    Moreover,
    if there is a marking $M$ 
    such that
    $M_{\vec{v}}\step{*} M$
    and $0\in \fractof{M}$,
    where $M$ contains a token of age $0$,
    then $M\le M_{\vec{v'}}$ for some encoding $M_{\vec{v}'}$ of a vector $\vec{v'}=F^m(\vec{v})$.
    This means that
    it suffices to add one transition $t$ with $\Pre(t)=(\True{0},0)$
    whose enabledness witnesses the existence of a reachable encoding $M_{\vec{v}'}$
    containing a token $(\True{0},0)$.
    By the properties above, there exists $m\in\N$ with $F^m(\vec{v})=\vec{v}'$ and 
    $\vec{v}'[0]=1$ iff $M_{\vec{v}}\step{*}M_{\vec{v}'}\step{t}$.
    \qedhere
\end{proof}
This lower bound holds even for discrete time TPN, e.g.~\cite{EFR2000},
because the proof uses only timed steps with duration $d=1$.

\section{Upper Bound}
\label{sec:upperbound}
We start by observing that we can restrict ourselves, without loss of
generality, to non-consuming TPN (\cref{def:nonconsuming}) for showing the
upper bound. Intuitively, since we start with an arbitrarily high number of
tokens anyway, it does not matter how many of them are consumed by transitions
during the computation, since some always remain.

\begin{restatable}{lemma}{lemnonconsuming}
    \label{lem:mtpn:wlog}
    \label{lem:nonconsuming}
    The \ECOVER\ problem for TPN logspace-reduces to the \ECOVER problem for non-consuming TPN. That is,
    for every TPN $\?N$ 
    and for every place $p$ and transition $t$ of $\?N$,
    one can
    construct, using logarithmic space,
    a non-consumimg TPN $\?N'$ 
    together with a place $p'$ and transition $t'$
    of $\?N'$,
    so that
    there exists $M\in\Coverset[\?N]{\N\cdot\{(p,{0})\}}$ enabling $t$ in $\?N$ if and
    only if there exists $M'\in\Coverset[\?N']{\N\cdot\{(p',0)\}}$ that enables $t'$ in $\?N'$.
\end{restatable}
\begin{proof}
    First notice that 
    the first condition in \cref{def:nonconsuming}, that asks that
    every transition takes at most one token each place, is merely a syntactic convenience.
    A net satisfying this condition can be constructed by adding a few extra places and intermediate transitions
    to first distribute tokens to those extra places for the original transition to consume.

    So let's assume w.l.o.g., that $\?N$ satisfies this condition
    and let
    $\?N'$ be the non-consuming variant
    derived from $\?N$ where for all transitions $T$,
    $\Post_{\?N'}(t)\eqdef\Post_\?N(t)\oplus\Pre_\?N(t)$.
    Notice that then, for every discrete step
    $M \dstep{t}M'$ we have that $M \leq M'$.
    We prove the following claim.

    \begin{claim}
        \emph{
    For every place $p$ and transition $t$ of $\?N$ there exists
    $M\in\Reachset[\?N]{\N\cdot\{(p,{0}\}}$ enabling $t$ in $\?N$ if, and
    only if there exists $M'\in\Reachset[\?N']{\N\cdot\{(p,0)\}}$ that enables $t$ in $\?N'$.
      }
    \end{claim}
    
    The ``$\?N\to\?N'$'' direction follows from the observation that
    the pointwise ordering $\le$ on 
    markings, is a simulation:
    If $M\step{}N$ and $M'\ge M$ then there exists an $N'\ge N$ with $M'\step{}N'$.
    For the other direction,
    suppose there exists a witnessing path
    \begin{equation*}
        m\cdot\{(p,{0})\}~=~M_0\step{}M_1\step{}M_2\step{}\cdots\step{}M_k\step{t}
    \end{equation*}
    of length $k$ in $\?N'$.
    We can inductively derive a witnessing path in $\?N$ backwards, again using the fact that $\le$ is a simulation.
    First note that if $M'$ enables $t$, then every $m'\cdot M'$ with $m'>0$ enables $t$, (in both nets).
    Suppose $M_i\step{\rho}$ is a path of length $(k-i)$ that ends in a $t$-transition.
    By the simulation property, there is such a path from every $m\cdot M_i$, $m>0$.
    Further, there must exist markings $M'_{i-1}\in\ \down(\N\cdot M_{i-1})$ and $M'_{i}\in\ \down(\N\cdot M_{i})$ such that
    $M'_{i-1}\step{} M'_{i}$.
    It suffices to pick $M'_{i-1}\eqdef B\cdot M_{i-1}$, where $B\in\N$ is the maximal
    cardinality of any multiset $\Pre(t)$
    (This number is itself bounded by $\card{P}\cdot \card{\Var}$ by our assumption on $\Pre(t)$).
    We conclude that in $\?N$ there is a path ending in a $t$-transition
    and starting in marking $(B\cdot k) \cdot M_0$, which is in $\N\cdot\{(p,{0})\}$.
   \qedhere 
\end{proof}

\subsection{Region Abstraction}
\label{sec:regionabs}

We recall a constraint system called regions defined for timed
automata \cite{AD:1994}. The version for TPN used here is similar to the one
in \cite{AMM07}.

Consider a fixed, nonconsuming TPN
$\?N=(P,T,\Var,G,\Pre,\Post)$.
Let $\cmax$ be the largest finite value appearing in transition guards $G$.
Since different tokens with age $>\cmax$ cannot be distinguished by transition guards, we consider only token ages below or equal to $\cmax$
and treat the integer parts of older tokens as equal to $\cmax +1$.
Let 
$\intof{c}\eqdef \min\{\cmax+1, \lfloor{c}\rfloor\}$
and $\fractof{c} \eqdef c-\lfloor{c}\rfloor$
for a real value $c\in\nnreals$.
We will work with an abstraction of TPN markings
as words over the alphabet $\Sigma \eqdef 2^{P\x\zeroto{\cmax+1}}$.
Each symbol $X\in\Sigma$ represents the places and integer ages of tokens for a particular fractional value.

\begin{definition}\label{oneclock:abstr}
    Let $M\subseteq P \x \nnreals$ be a marking
    and let $\fractof{M} \eqdef \{\fractof{c}\mid (p,c)\in M\}$ be the set of fractional clock values that appear in $M$.

    Let $S\subset [0,1[$ be a finite set of real numbers
    with $0\in S$ and $\fractof{M}\subseteq S$
    and let $f_0,f_1,\dots,f_n$, be an enumeration of $S$ so that $f_{i-1} < f_i$ for all $i\le n$.
    The \emph{$S$-abstraction} of $M$ is 
    $$\abstr[S]{M} \eqdef x_0x_1\dots x_n \in \Sigma^*$$
    where
    $x_i \eqdef \{(p,\intof{c}) \mid (p,c)\in M \land \fractof{c}=f_i\}$
    for all $i\le n$.
    We simply write $\abstr{M}$ for the shortest abstraction, i.e. with respect to
    $S=\{0\}\cup\fractof{M}$.
\end{definition}

\begin{example}
    The abstraction of marking $M = \{(p,2.1),(q,2.2),(p,5.1),(q,5.1)\}$
    is $\abstr{M}=\emptyset~\{(p,2),(p,5),(q,5)\}~\{(q,2)\}$. The first symbol is $\emptyset$,
    because $M$ contains no token with an integer age (i.e., no token whose age
    has fractional part $0$).
    The second and third symbols represent sets of tokens with fractional values $0.1$ and $0.2$, respectively.
\end{example}

Clocks with integer values play a special role in the behavior of TPN, because
the constants in the transition guards are integers. Thus we
always include the fractional part $0$ in the set $S$ in \cref{oneclock:abstr}.

We use a special kind of regular expressions over $\Sigma$ to represent coverable sets of TPN markings as follows.

\begin{definition}\label{expression:def}
    A regular expression $E$ over $\Sigma$ represents the downward-closed set of TPN markings
    covered by one that has an abstraction in the language of $E$:
    $$\denotationof{E}\eqdef \{N\mid \exists M \exists S.~
    M\ge N \land \abstr[S]{M}\in\Lang{E}\}.$$

    An expression is \emph{simple} if it is of the form
    $E = x_0x_1\dots x_{k}$ where for all $i\le k$
    either $x_i\in\Sigma$ or $x_i={y_i}^*$ for some $y_i\in\Sigma$.
    In the latter case we say that $x_i$ \emph{carries a star}.
    That is, a simple expression is free of Boolean combinators
    and uses only concatenation and Kleene star.
    We will write $\hat{x}_i$ to denote the 
    symbol in $\Sigma$ at position $i$: it is $x_i$ if
    $x_i\in\Sigma$ and $y_i$ otherwise.
\end{definition}

\begin{remark}
Notice that for all simple expressions $\alpha,\beta$ so that $\len{\alpha}>0$,
we have that $\denotationof{\alpha\emptyset\beta} =\denotationof{\alpha\beta}$.
However, unless $\alpha$ has length $0$ or is of the form $\alpha=\emptyset\alpha'$, we have
$\denotationof{\emptyset\alpha} \neq\denotationof{\alpha}$.
This is because a marking $M$ that contains a token $(p,c)$
with $\fractof{c}=0$ has the property that all abstractions
$\abstr[S]{M}=x_0\dots x_k$ of $M$ have $x_0\neq \emptyset$.
\end{remark}

The following lemmas express the effect of TPN transitions at the level of the region abstraction.
\Cref{lem:discall,lem:simulation} state that maximally firing of discrete transitions (the relation $\NDStep{*}$) is computable and monotone.
\Cref{lem:epsilonsteps,lem:nonstarrotation} state how to represent timed-step successor markings.

\begin{restatable}{lemma}{lemdiscall}
\label{lem:discall}
\label{lem:monotonicity}
For every non-consuming TPN $\?N$ there are polynomial time computable functions
$f: \Sigma \times \Sigma \times \Sigma \to \Sigma$
and 
$g: \Sigma \times \Sigma \times \Sigma \to \Sigma$ with the following
properties.
    \begin{enumerate}
    \item\label{lem:discall-ad1}
      $f$ and $g$ are monotone (w.r.t.\ subset ordering) in each argument.
    \item\label{lem:discall-ad2} 
      $f(\alpha, \beta, x) \supseteq x$ and $g(\alpha, \beta, x) \supseteq x$ for all $\alpha,\beta,x \in \Sigma$.
    \item\label{lem:discall-ad3}
      Suppose that $E=x_0x_1\dots x_k$ is a simple expression, 
      $\alpha \eqdef x_0$ and
      $\beta \eqdef \bigcup_{i>0} \hat{x}_i$, and
      $E' = x'_0x'_1\dots x'_k$ is the derived expression defined by conditions:
      \begin{enumerate}
          \item $x_0' \eqdef f(\alpha,\beta,x_0)$,
          \item $x_i' \eqdef g(\alpha,\beta,\hat{x}_i)^*$ for $i>0$,
          \item $x_i'$ carries a star iff $x_i$ does.
      \end{enumerate}
      Then $\denotationof{E'} = \{M'' \mid \exists M \in \denotationof{E} \land M\Dstep{*}M' \ge M''\}$.
\end{enumerate}
\end{restatable}
A proof of this statement is in the appendix.
It is essentially due to the monotonicity of discrete transition firing in TPN and the fact that iteratively firing transitions must saturate due to the nonconsuming semantics.
We first prove it only for star-free expressions $E$ in condition 3 (\cref{lem:discall-starfree}) and then generalize to all simple expressions by induction.

\begin{definition}\label{def:discall}
We will write $\SAT{E} \eqdef E'$ for the successor expression $E'$ of $E$
guaranteed by \cref{lem:discall}. I.e., $\SAT{E}$ is the
saturation of $E$ by maximally firing discrete transitions.
\end{definition}

Notice that by definition it holds that $\denotationof{E} \subseteq \denotationof{\SAT{E}}\subseteq\Coverset{\denotationof{E}}$,
and consequently also that $\Coverset{\denotationof{\SAT{E}}}=\Coverset{\denotationof{E}}$.

\begin{lemma}
    \label{lem:simulation}
    Suppose that $X=x_0x_1\dots x_k$ is a simple expression of length $k+1$ with
    $\SAT{X}=x'_0x'_1\dots x'_k$ 
    and $x_0,x'_0\in\Sigma$.
    Let $Y=y_0\alpha_1y_1\alpha_2\dots \alpha_k y_k$ 
    be a simple expression
    with
    $\SAT{Y}=y'_0\alpha'_1y'_1\alpha'_2\dots \alpha'_k y'_k$
    and $y_0,y'_0\in\Sigma$.

    If 
    $\hat{x}_i \subseteq \hat{y}_i$
    for all $i\le k$
    then
    $\hat{x}'_i 
    \subseteq
    \hat{y}'_i 
    $
    for all $i\le k$.
\end{lemma}
\begin{proof}
    The assumption of the lemma provides that
    $\alpha_x\eqdef x_0 \subseteq \alpha_y\eqdef y_0$ and
    $\beta_x\eqdef \bigcup_{k\ge i>0}\hat{x}_i \subseteq \beta_y\eqdef \bigcup_{k\ge i>0}\hat{y}_i$.
    Therefore, by \cref{lem:discall-ad1} of \cref{lem:discall},
    we get that
    $$
    x'_0 = f(\alpha_x,\beta_x,x_0)
    \quad\subseteq\quad
    f(\alpha_y,\beta_y,y_0)
    =y'_0
    $$
    and similarly, for all $k\ge i\ge 0$, that
    $\hat{x}'_i = g(\alpha_x,\beta_x,\hat{x}_i)
    ~\subseteq~
    g(\alpha_y,\beta_y,\hat{y}_i)
    = \hat{y}'_i.
    $
\end{proof}

For $x\in\Sigma$ we write $(x+1)\eqdef\{(p,\intof{n+1})\mid (p,n)\in x\}$ for the symbol where token ages are incremented by $1$.
\begin{restatable}{lemma}{lemepsilonsteps}
    \label{lem:epsilonsteps}
    $
    \denotationof{\emptyset E} = 
    \{M' \mid \exists M \in \denotationof{E} \land M\tstep{d}M' \land d < 1-\max(frac(M))\}
    $.
\end{restatable}
\begin{proof}
\emph{``$\supseteq$''}:
Suppose that
$M$ is a non-empty marking in $\denotationof{E}$, 
$d<1-\max(\fractof{M})$ and
$M\tstep{d}M'$.
The assumption on $d$ implies that
for every token $(p,c) \in M$ we have $\intof{c}=\intof{c+d}$. In other words, 
the integral part of the token age remained the same.
Therefore $(p,\intof{c}) = (p,\intof{c+d}) \in M'$.
Also from the assumption on $d$ we get that
$$\fractof{M'}= \{x+d \mid x \in \fractof{M}\}$$
Recall that $\abstr{M}=\abstr[S]{M}$
and $\abstr{M'}=\abstr[S']{M'}$
for the sets $S\eqdef\{0\}\cup \fractof{M}$
and $S'\eqdef\{0\}\cup \fractof{M'}$.
Clearly, $0\notin \fractof{M'}$.
There are two cases:

\begin{enumerate}
    \item 
$0\in\fractof{M}$.
Then $\abstr{M'}=\emptyset\abstr{M}\in \Lang{\emptyset E}$,
and consequently, 
$M'\in\denotationof{\emptyset E}$.

\item $0\notin\fractof{M}$.
Then $\abstr{M'}=\abstr{M}=\emptyset w\in \Lang{E}$.
Suppose that $E=x_0\alpha$, i.e., $E$ has $x_0\in\Sigma$ as its leftmost symbol,
and $w\in\Lang{\alpha}$.
If $x_0=\emptyset$ then
$\denotationof{E}=\denotationof{\emptyset E}$ and thus $\abstr{M'}\in\denotationof{\emptyset E}$.
Otherwise, if $x_0\neq\emptyset$
then $x_0w\in\Lang{E}$ and $x_0w=\abstr{M''}$ for some marking $M''\ge M'$.
So again, $M'\in\denotationof{\emptyset E}$.
\end{enumerate}

\emph{``$\subseteq$''}:
W.l.o.g., pick a non-empty marking $M'\in \denotationof{\emptyset E}$.
If $E$ has $\emptyset$ as its leftmost symbol, then
$\denotationof{\emptyset E}=\denotationof{E}$
and the claim follows using $d=0$, since then $M'\in\denotationof{E}$.
So suppose that $E$ does not start with $\emptyset$.
Note that by \cref{oneclock:abstr},
there are no tokens in the marking $M'$ whose clocks have fractional value zero.
Let $$d\eqdef\min(\fractof{M'})$$ be the minimal fractional clock value among the tokens of $M'$
and based on this, define
$M \eqdef \{ (p,c-d)\mid (p,c)\in N' \} $.  
By construction of $M$ we get $M\tstep{d}M'$
and also that $\max(\fractof{M}) = \max(\fractof{M'})-d < 1$.
Therefore that $1-\max(\fractof{M}) < 1 -d$.
Finally, observe that
$\fractof{M}= \{x-d \mid x \in \fractof{M'} \}$
and $0\in\fractof{M}$.
It follows that $\abstr{M'}=\emptyset\abstr{M}$
and therefore that $\abstr{M}\in \Lang{E}$
and $M\in\denotationof{E}$.
This means that $M'$ is included in the set on the right in the claim.
\end{proof}
\begin{restatable}{lemma}{lemnonstarrotation}\label{oneclock:rota}
    \label{lem:nonstarrotation}
    Let $\alpha z$ be a simple expression
    where $\hat{z}=z\in\Sigma$
    (the rightmost symbol is not starred).
    Then, 
    $\denotationof{(z+1)\alpha}$ contains a marking $N$
    if, and only if,
    there exists markings $N'\ge N$ and $M$,
    and a set $S\subseteq [0,1[$ so that
    \begin{enumerate}
        \item $\card{S} = \len{\alpha z}$
        \item $\abstr[S]{M}\in \Lang{\alpha z}$
        \item $M\tstep{d}N'$ for $d = 1-\max(S)$.
    \end{enumerate}
\end{restatable}
\begin{proof}
    Suppose markings $N,N',M$, a set $S\subseteq [0,1[$ and $d\in\nnreals$ so that the conditions 1 to 3 are satisfied.
    Let $S'\eqdef\{0\}\cup \{s+d\mid s\in S\setminus\{d\}\}$.
    Then, $\card{S'}=\card{S}$ and 
    $\abstr[S']{N'} \in\Lang{(z+1)\alpha}$,
    which witnesses that
    $N\in\denotationof{(z+1)\alpha}$.

    Conversely, let $N\in\denotationof{(z+1)\alpha}$
    be a non-empty marking.
    If $\len{\alpha} = 0$,
    then $N\in\denotationof{(z+1)}$
    and so
    $\abstr[S]{N}\in\Lang{(z+1)}$
    for $S\eqdef\fractof{N}=\{0\}$.
    This means that
    $M\tstep{1}N=(M+1)$ for a marking $M$ with
    $\abstr[S]{M}\in\Lang{z}=\Lang{\alpha z}$.

    If $\len{\alpha} >0$,
    pick some marking $N'\ge N$ and set $S'$ so that
    $\abstr[S']{N'}=(z+1)w$, for some word $w\in\Lang{\alpha}$.
    Then we must have that $\card{S'}=\len{(z+1)\alpha}>1$
    and so $d\eqdef \min(S'\setminus\{0\})$ exists.
    Let $S\eqdef\{s-d\mid s\in S'\}\cup\{1-d\}$
    and $M$ be the unique marking with $M\tstep{d}N'$.
    Notice that $1-d=\max(S)$.
    It follows that $\abstr[S]{M}=wz\in\Lang{\alpha z}$.
\end{proof}

\bigskip
\noindent
We will often use the following simple fact, which is a direct consequence of \cref{lem:nonstarrotation}.
\begin{corollary}
    \label{cor:nonstarrotation}
    $\denotationof{(z+1)\alpha}
    \subseteq \Coverset{\denotationof{\alpha z}}$.
\end{corollary}

\medskip
\noindent
Finally, the following lemma will be the basis for our exploration algorithm.
\begin{restatable}{lemma}{lemaccstep}
    \label{lem:accstep}
    Let $\alpha x_0^*$ be a simple expression with
    $\SAT{\alpha x_0^*}=\alpha x_0^*$.
    Then
    $
    \Coverset{\denotationof{\alpha x_0^*}} = 
    \denotationof{\alpha x_0^*}
    \cup \Coverset{\denotationof{(x_0+1)\alpha x_0^*}}$.
\end{restatable}
\begin{proof}
    For the right to left inclusion notice that
    $\denotationof{\alpha x_0^*} \subseteq \Coverset{\denotationof{\alpha x_0^*}}$ trivially holds.
    For the rest, we have
    $
    \denotationof{(x_0+1)\alpha x_0^*}
    \subseteq
    \Coverset{\denotationof{\alpha x_0^*}}
    $ by \cref{cor:nonstarrotation}, and therefore
    $\Coverset{\denotationof{(x_0+1)\alpha x_0^*}}
    ~\subseteq~
    \Coverset{\Coverset{\denotationof{\alpha x_0^*}}}
    =
    \Coverset{\denotationof{\alpha x_0^*}}
    $.
    For the left to right inclusion, we 
    equivalently show that
    \begin{equation}
    \Coverset{\denotationof{\alpha x_0^*}} \setminus
    \denotationof{\alpha x_0^*}
    \subseteq
    \Coverset{\denotationof{(x_0+1)\alpha x_0^*}}
    \end{equation}
    Using the assumption that $\SAT{\alpha x_0^*} = \alpha x_0^*$, the set on the left contains everything coverable from $\denotationof{\alpha x_0^*}$ by a sequence that starts with a (short) time step.
    It can therefore be written as
    $$\Coverset{
    \{N_1\mid \exists N_0\in\denotationof{\alpha x_0^*} \land N_0\tstep{d}N_1
        \land 0<d< 1-\max(frac(N_0))
    \}
    }.
    $$
    By \cref{lem:epsilonsteps} and because $\denotationof{\emptyset\alpha}\subseteq \denotationof{X\alpha}$ for all $X\in \Sigma$ and $\alpha\in\Sigma^*$, we conclude that indeed,
    $
    \Coverset{\denotationof{\alpha x_0^*}} \setminus
    \denotationof{\alpha x_0^*}
    ~\subseteq~
    \Coverset{\denotationof{\emptyset\alpha x_0^*}}
    \subseteq
    \Coverset{\denotationof{(x_0+1)\alpha x_0^*}}
    $.
\end{proof}

\subsection{Acceleration}

We propose an acceleration procedure based on unfolding expressions according to \cref{lem:accstep} (interleaved with saturation steps to guarantee its premise) and introducing new Kleene stars to keep the length of intermediate expressions bounded.
This 
procedure (depicted in \cref{alg:inner_loop}),
is used to characterize an initial subset of the coverability set.

\renewcommand{\algorithmicrequire}{\textbf{Input:}}
\renewcommand{\algorithmicensure}{\textbf{Output:}}
\begin{algorithm}[th]
    \begin{algorithmic}[1]
        \Require a simple expression $S_0 = x_1x_0^*$ (of length 2 and with last symbol starred)
        \Ensure simple expressions $S_1, S_i$ and $R$,
        of lengths 2, 4, and 2, respectively.
   	
  \State $S_1 \eqdef x_1^1(x_0^1)^*= \SAT{x_1x_0^*}$
  \State $S_2 \eqdef x_2^2x_1^2(x_0^2)^*= \SAT{(x_0^1+1)S_1}$
  \State $S_3 \eqdef x_3^3x_2^3x_1^3(x_0^3)^*= \SAT{(x_0^2+1)S_2}$
  \State $i\gets 3$

  \Repeat
         \State $x_{i+1}^{i+1}x_{i}^{i+1}x_{i-1}^{i+1}x_1^{i+1}(x_0^{i+1})^*\eqdef \SAT{(x_0^{i}+1)S_{i}}$
         \State $S_{i+1} \eqdef x_{i+1}^{i+1}(x_{i}^{i+1})^*x_1^{i+1}(x_0^{i+1})^*$
  \State $i\gets i+1$
  \Until{$S_{i}=S_{i-1}$}
  \State $R \eqdef (x_1^i+1)(x_{i-1}^i)^*$
\State \Return $S_1, S_i, R$
	\end{algorithmic}
	
    \caption{Accelerate}
    \label{alg:inner_loop}
 \end{algorithm}

\begin{figure}[H]
\begin{center}
\begin{tikzpicture}[node distance=1.05 and 1.4, on grid]

\def\rdist{3.25cm}
\def\ldist{7.25cm}

\tikzstyle{graphnode} = [
align=center]
\tikzstyle{newnode} = [graphnode, red] 
\tikzstyle{nnewnode} =   [graphnode, black] 
\tikzstyle{snode} =   [graphnode, black] 
\tikzstyle{mnode} =   [graphnode, black] 
\tikzstyle{label} =   [align=left,text width=3.5cm]
\tikzstyle{f}  = [->, lightgray]
\tikzstyle{eq}  = [-, lightgray]
\tikzstyle{dinc} = [->, red]
\tikzstyle{ddinc} = [draw,->, blue]

\node (S0-0) [snode] {$x_0^{*}$};
\node (S0-1)[snode, left=of S0-0] {$x_1$};
\node (S0-label) [label,right=\rdist of S0-0] {start};

\node (S1-0)  [snode, below =of S0-0] {$(x_0^1)^{*}$};
\node (S1-1) [ snode, left =of S1-0] {$x_1^1$};

\node[label] (S1-label) [below=of S0-label] {$S_1 = \SAT{x_1x_0^*} $} ;

 \node (S1-rot-0) [snode, below =of S1-0] {$(x_0^1)^{*} $};
 \node (S1-rot-1) [snode, left =of S1-rot-0] {$x_1^1$};
 \node (S1-rot-2) [newnode, left =of S1-rot-1] { $(x_0^1+1)$}; 
 \node (S1-rot-label) [label,below =of S1-label] {$(x_0^1+1)S_1 $}; 
 
 \node (S2-0) [snode, below =of S1-rot-0] {$(x_0^2)^{*}$};
 \node (S2-1) [snode, left =of S2-0 ] {$x_1^2$};
 \node (S2-2) [nnewnode, left = of S2-1] {$x_2^2$};
 \node (S2-label)[label, below=of S1-rot-label] {$S_2=\SAT{(x_0^1+1)S_1}$};
 
 \node (S2-rot-0) [snode, below =of S2-0  ] {$(x_0^2)^{*} $};
 \node (S2-rot-1) [snode, left = of S2-rot-0] {$x_1^2$};
 \node (S2-rot-2) [mnode, left = of S2-rot-1] { $x_2^2$};
 \node (S2-rot-3)[newnode, left =of S2-rot-2]{ $(x_0^2+1)$};
 \node (S2-rot-label) [label,below =of S2-label] {$(x_0^2+1)S_2 $}; 
 
 \node (S3-0) [snode, below =of S2-rot-0]{$(x_0^3)^{*}$};
 \node (S3-1) [snode, left =of S3-0 ] {$x_1^3$};
 \node (S3-2) [mnode, left =of S3-1 ]{$x_2^3$};
 \node (S3-3)  [nnewnode, left =of S3-2 ]{$x_3^3$};
 \node (S3-label)[label, below=of S2-rot-label] {$S_3=\SAT{(x_0^2+1)S_2}$};
  
 \node (S3-rot-0) [snode, below =of S3-0  ] {$(x_0^3)^{*} $};
 \node (S3-rot-1) [snode, left =of S3-rot-0] {$x_1^3$};
 \node (S3-rot-2) [mnode, left =of S3-rot-1] { $x_2^3$};
 \node (S3-rot-3) [mnode,left =of S3-rot-2] {$x_3^3$};
 \node (S3-rot-4)[newnode, left =of S3-rot-3]{ $(x_0^3+1)$};
 \node (S3-rot-label) [label,below =of S3-label] {$(x_0^3+1)S_3 $}; 
 
 \node (L4-0) [snode,below =of S3-rot-0]{$(x_0^4)^{*}$};
 \node (L4-1) [snode, left =of L4-0]{$x_1^4$};
 \node (L4-2) [mnode,left =of L4-1] {$x_2^4$};
 \node (L4-3) [mnode,left =of L4-2] {$x_3^4$};
 \node (L4-4) [nnewnode,left =of L4-3] {$x_4^4$};
 \node (L4-label)[label, below=of S3-rot-label] {$\SAT{(x_0^3+1)S_3}$};
 
 \node (S4-0) [snode, below=of L4-0] {$(x_0^4)^{*}$};
 \node (S4-1) [snode, left =of S4-0]{$x_1^4$};
 \node (S4-2) [mnode, left=of S4-1]{};
 \node (S4-3) [mnode, left=of S4-2]{$(x_3^4)^{*}$};
 \node (S4-4) [nnewnode, left =of S4-3]{$x_4^4$};
 \node (S4-label)[label, below=of L4-label] {$S_4$};

 \node (S4-rot-0) [snode, below =of S4-0]{$(x_0^4)^{*}$};
 \node (S4-rot-1) [snode, left = of S4-rot-0] {$x_1^4$   };
 \node (S4-rot-2) [mnode, left=of S4-rot-1]{};
 \node (S4-rot-3) [mnode, left=of S4-rot-2]{$(x_3^4)^{*}$  };
 \node (S4-rot-4)[mnode, left =of S4-rot-3]{$x_4^4$};
 \node (S4-rot-5) [newnode, left =of S4-rot-4]{$(x_0^4+1)$  };
 \node (S4-rot-label) [label,below =of S4-label] {$(x_0^4+1)S_4 $}; 
 
 \node (L5-0) [snode,  below=of S4-rot-0] {$(x_0^5)^{*} $};
 \node (L5-1) [snode,  left=of L5-0] {$x_1^5$};
 \node (L5-2) [mnode,  left=of L5-1] {};
 \node (L5-3) [mnode,  left=of L5-2] {$(x_3^5)^*$};
 \node (L5-4) [mnode,  left=of L5-3] {$x_4^5$};
 \node (L5-5) [nnewnode,  left=of L5-4] {$x_5^5$};
 \node (L5-label)[label, below=of S4-rot-label] {$\SAT{(x_0^4+1)S_4}$};

 \node (S5-0) [snode, below=of L5-0] {$(x_0^5)^{*}$};
 \node (S5-1) [snode,  left=of S5-0] {$x_1^5$};
 \node (S5-2) [snode,  left=of S5-1] {};
 \node (S5-3) [snode,  left=of S5-2] {};
 \node (S5-4) [mnode,  left=of S5-3] {$(x_4^5)^{*}$};
 \node (S5-5) [nnewnode,left=of S5-4] {$x_5^5$};
 \node (S5-label)[label, below=of L5-label] {$S_5$};

 \node (S5-rot-0) [snode, below=1cm of S5-0]{$\vdots$};
 \node (S5-rot-1) [snode, left = of S5-rot-0] {$\vdots$};
 \node (S5-rot-2) [mnode, left=of S5-rot-1]{};
 \node (S5-rot-3) [mnode, left=of S5-rot-2]{};
 \node (S5-rot-4)[mnode, left =of S5-rot-3]{$\vdots$};
 \node (S5-rot-5)[mnode, left =of S5-rot-4]{$\vdots$};
 \node (S5-rot-6) [newnode, left =of S5-rot-5]{$\vdots$};
 \node (S5-rot-label) [label,right=\rdist of S5-rot-0] {$\vdots$}; 

\draw[f](S0-0) -> (S1-0);
\draw[f](S0-1) -> (S1-1);
\foreach \x[evaluate={\lx=int(\x-1);\llx=int(\x-2);}] in {2,...,5}{
  \foreach[evaluate={\li=int(\i-1);\lli=int(\i-2);}] \i in {0,...,\x}{
    \ifthenelse{\x<4}{
        \draw[f](S\lx-rot-\i) -> (S\x-\i);
        }{
            \ifthenelse{\i>\llx \OR \i=\llx \OR \i=0 \OR \i=1}{
              \draw[f](S\lx-rot-\i) -> (L\x-\i);
            }{}
        }
  }
} 
\draw[dinc](S1-rot-2) -> (S2-rot-3);
\draw[dinc](S2-rot-3) -> (S3-rot-4);
\draw[dinc](S3-rot-4) -> (S4-rot-5);
\draw[dinc](S4-rot-5) -> (S5-rot-6);

\draw[ddinc](S2-2) -> (S3-3);
\draw[ddinc](S3-3) -> (S4-4);
\draw[ddinc](S4-4) -> (S5-5);

\draw[ddinc](S3-2) -> (S4-3);
\draw[ddinc](S4-3) -> (S5-4);

\foreach \x[evaluate={\lx=int(\x-1);\llx=int(\x-2);}] in {1,...,5}{
  \foreach[evaluate={\li=int(\i-1);\lli=int(\i-2);}] \i in {0,...,\x}{
    \ifthenelse{\i>\llx \OR \i=0 \OR \i=1}{
        \ifthenelse{\x>3}{
        \draw[eq](L\x-\i) -> (S\x-\i);
        \draw[eq](S\x-\i) -> (S\x-rot-\i);
        }{
        \draw[eq](S\x-\i) -> (S\x-rot-\i);
        }
    }{}
  }
} 

\node[label, left=\ldist of S0-0] {line};
\node[label, left=\ldist of S1-0] {1:};
\node[label, left=\ldist of S1-rot-0] {2:};
\node[label, left=\ldist of S2-rot-0] {3:};
\node[label, left=\ldist of S3-rot-0] {6:};
\node[label, left=\ldist of S4-0] {7:};
\node[label, left=\ldist of S4-rot-0] {6:};
\node[label, left=\ldist of S5-0] {7:};

\end{tikzpicture}

\end{center}
\caption{A Run of \cref{alg:inner_loop} (initial steps).
The column on the left indicates the line of code,
the middle depicts the current expression and the column on the right recalls its origin.
Gray bars indicate that the respective symbols are equal. Arrows denote (set) inclusion between symbols.
The gray vertical arrows indicate inclusions due to saturation (\cref{lem:discall}),
as claimed in item~1 of \cref{lem:termobs}.
Red and blue arrows indicate derived inclusions (as stated in \cref{lem:termobs}).
}
\label{fig:algorithm1}
\end{figure}
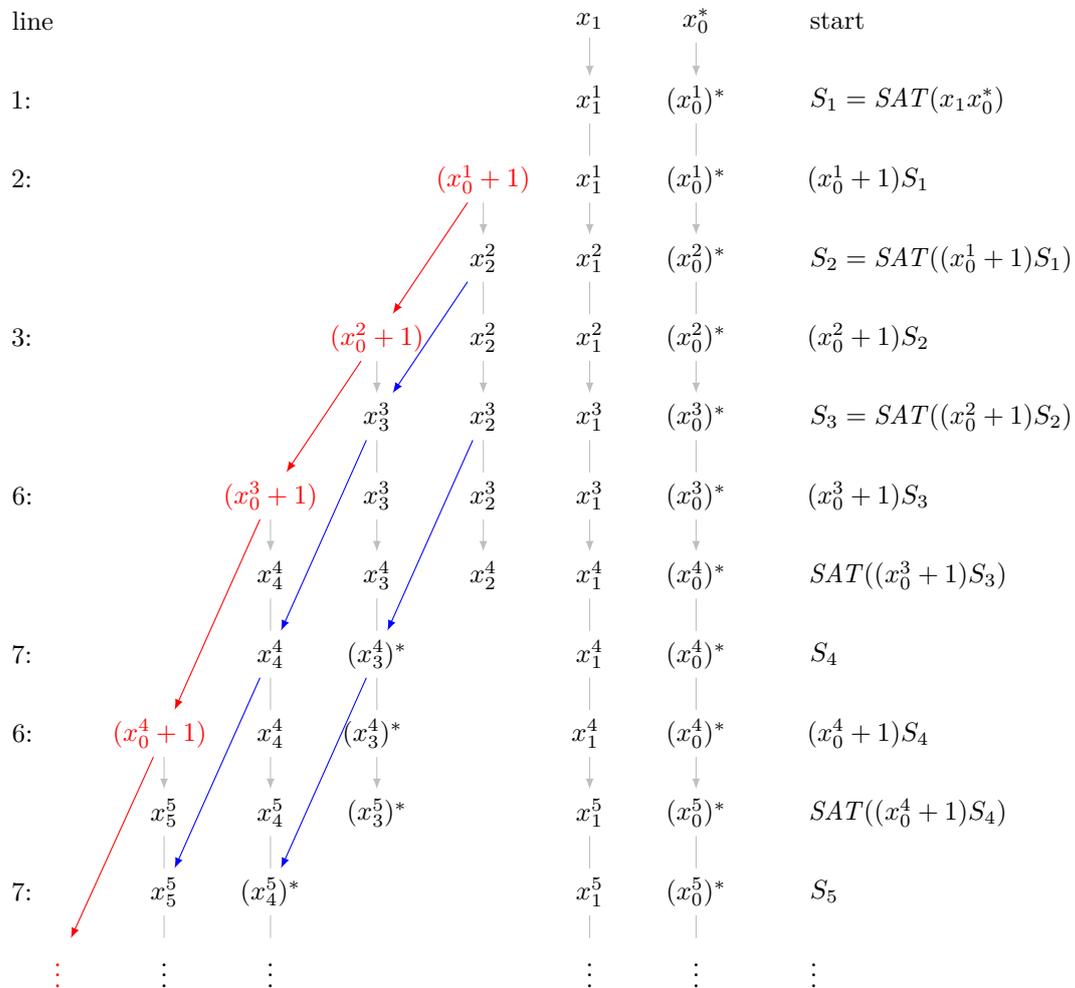

Given a length-2 simple expression $S_0$ where the rightmost symbol is starred, the algorithm will 
first saturate (\cref{def:discall}, in line 1),
and then alternatingly rotate a copy of the rightmost symbol (\cref{oneclock:rota}), and saturate the result
(see lines 2, 3, 6).
Since each such round extends the length of the expression by one,
we additionally collapse them (in line 7) by adding an extra Kleene star to the symbol at the second position.
The crucial observation for the correctness of this procedure is that the subsumption step in line 7
does not change the cover sets of the respective expressions.

Observe that \cref{alg:inner_loop} is well defined because
the $\SAT{S_i}$ are computable by \cref{lem:discall}.
Termination is guaranteed by the following simple observation.

\begin{lemma}
    \label{lem:termobs}
    Let $x_j^i\in\Sigma$ be the symbols computed by
    \cref{alg:inner_loop}. Then
    \begin{enumerate}
        \item $x_{j}^{i+1}\supseteq x_{j}^{i}$, for all $i>j\ge 0$.
        \item $x_{i}^{i}\supseteq x_{i-1}^{i-1}$
            and 
            $x_{i}^{i+1}\supseteq x_{i-1}^{i}$,
            for all $i\ge 3$.
    \end{enumerate}
\end{lemma}
\begin{proof}
    The first item is guaranteed by Point~2 of \cref{lem:discall}.
    In particular this means that 
$x_0^{i+1}\supseteq x_0^{i}$ 
and therefore that
$(x_0^{i+1}+1)\supseteq (x_0^{i}+1)$
for all $i\ge 0$ (indicated as red arrows in \cref{fig:algorithm1}). The second item now follows from this observation by \cref{lem:simulation}.
\end{proof}

\begin{lemma}[Termination]
    \label{lem:termination}
    \cref{alg:inner_loop} terminates
    with $i \le 4\cdot \card{P}\cdot(\cmax+1)$.
\end{lemma}
\begin{proof}
From \cref{lem:termobs} we deduce that for all $i\ge2$, the expression $S_{i+1}$ is point-wise larger than or equal to $S_i$ with respect to the subset ordering on symbols.
The claim now follows from the observation that
all expressions $S_{i\ge 3}$ have length $4$
and that every symbol $x_i\in\Sigma$ can only increase
at most $\card{P}\cdot (\cmax+1)$ times.
\end{proof}

\begin{lemma}[Correctness]
    \label{lem:acc:correctness}
    Suppose that $S_1,S_\ell, R$ be the expressions
    computed by \cref{alg:inner_loop} applied to 
    the simple expression $x_1x_0^*$.
    Then
    $\Coverset{\denotationof{x_1x_0^*}}
    = \denotationof{S_1}
    \cup \denotationof{S_\ell}
    \cup \Coverset{\denotationof{R}}
    $.
\end{lemma}
\begin{proof}
    Let $S_1,\ldots S_\ell$ denote the expressions defined in lines 1,2,3, and 7 of the algorithm. That is, $\ell$ is the least index $i$ such that $S_{i+1}=S_{i}$.
We define a sequence $E_i$ of expressions inductively,
starting with $E_1\eqdef S_1$
and if $E_i=e_i^ie_{i-1}^i\dots e_0^i$, we let
$
E_{i+1} \eqdef e_{i+1}^{i+1}e_i^{i+1}e_{i-1}^{i+1}\dots e_0^{i+1}
\eqdef \SAT{(\hat{e}_0^i+1)E_i}
$.
Here, the superscript indicates the position of a symbol
and not iteration.
This is the sequence of expressions resulting 
from unfolding \cref{lem:accstep}, interleaved with saturation steps, just in line 6 of the algorithm.
That is, the expressions $E_i$ are \emph{not} collapsed (line 7) and instead grow in length with $i$.
Still,
$E_1=S_1$,
$E_2=S_2$ and
$E_2=S_3$, but
$E_4\neq S_4$,
because the latter is the result of applying the subsumption step of line $7$ in our algorithm.
Notice that
$\Coverset{\denotationof{x_1x_0^*}} = \left(\bigcup_{k-1\ge i\ge 1}\denotationof{E_i}\right) \cup \Coverset{\denotationof{E_k}}$ holds
for all $k\in\N$.
We will use that
\begin{equation}
    \label{eq:acc:EisS}
\bigcup_{i\ge 2}\denotationof{E_i}
= \bigcup_{i\ge 2}\denotationof{S_i}
=\denotationof{S_\ell}.
\end{equation}
We start by observing that for all $i,j\in\N$ it holds that
$e_j^i = x_j^i$.
For $i\le 3$ this holds trivially by definition of $E_i=S_i$. For larger $i$, this can be seen by induction using \cref{lem:discall}.
Towards the first equality in \cref{eq:acc:EisS}, let $S_i^j$ be the expression resulting from
$S_i = x_{i}^{i}({x_{i-1}^{i}})^{*}x_{1}^{i}({x_0^{i}})^*$ by unfolding the first star $j$ times.
That is, $S_i^j\eqdef x_{i}^{i}({x_{i-1}^{i}})^{(j)}x_{1}^{i}(x_0^{i})^*$,
where the superscript $(j)$ denotes $j$-fold concatenation.
Clearly, $\denotationof{S_i}=\bigcup_{j\ge 0}\denotationof{S_i^j}$
and so the $\supseteq$-direction of the first equality in \cref{eq:acc:EisS} follows by
\begin{align*}
\denotationof{S_i^j}= \denotationof{x_{i}^{i}({x_{i-1}^{i}})^{(j)}x_{1}^{i}(x_0^{i})^*}
&\subseteq
\denotationof{x_{i+j}^{i+j}
\left(
{x_{i+j-1}^{i+j}}
{x_{i+j-2}^{i+j}}
\ldots
{x_{i}^{i+j}}
\right)
x_{1}^{i+1}(x_0^{i+1})^*}\\
&\subseteq
\denotationof{x_{i+j}^{i+j}
\left(
{x_{i+j-1}^{i+j}}
{x_{i+j-2}^{i+j}}
\ldots
{x_{i}^{i+j}}
\right)
\left(
{x_{i-1}^{i+j}}
\ldots
{x_{2}^{i+j}}
\right)
x_{1}^{i+1}(x_0^{i+j})^*}\\
&=\denotationof{E_{i+j}},
\end{align*}
where the first inclusion is due to
\cref{lem:termobs}.
The same helps for the other direction:
\begin{equation}
\denotationof{E_i} 
= \denotationof{x_{i}^{i}x_{i-1}^{i}x_{i-2}^{i}\dots x_2^ix_1^i x_0^{i}}
\subseteq \denotationof{x_{i}^{i}{(x_{i-1}^{i})}^{(i-2)}x_{1}^{i}x_0^{i}}
=\denotationof{S_i^{i-2}}
=\denotationof{S_i},
\end{equation}
which completes the proof of the first equality in \cref{eq:acc:EisS}.
The second equality holds because $\denotationof{S_i}\subseteq\denotationof{S_{i+1}}$
for all $i\ge 2$, by \cref{lem:termobs}, and by definition of $S_\ell=S_{\ell+1}$.
As a next step we show that
\begin{equation}
    \label{eq:acc:Rconnection}
\Coverset{\denotationof{S_\ell}}
= 
\denotationof{S_\ell}
\cup
\Coverset{\denotationof{R}}
\end{equation}
First observe that
$
\denotationof{R}
=\denotationof{(x_1^\ell+1){(x_{\ell-1}^\ell)}^*}
=\denotationof{(x_1^\ell+1)x_\ell^\ell {(x_{\ell-1}^\ell)}^*}
$
and consequently,

\begin{align*}
\Coverset{\denotationof{R}}
&=
\Coverset{\denotationof{(x_1^\ell+1)x_\ell^\ell {(x_{\ell-1}^\ell)}^*}}
\\
&\subseteq
\Coverset{\denotationof{x_\ell^\ell {(x_{\ell-1}^\ell)}^*x_1^\ell}}\\
&\subseteq
\Coverset{\denotationof{x_\ell^\ell {(x_{\ell-1}^\ell)}^*x_1^\ell {(x_0^\ell)}^*}}
=\Coverset{\denotationof{S_\ell}}
\end{align*}
where 
the first equation follows by
\cref{cor:nonstarrotation}
and the second because
$\Lang{x_\ell^\ell {(x_{\ell-1}^\ell)}^*x_1^\ell}\subseteq \Lang{x_\ell^\ell {(x_{\ell-1}^\ell)}^*x_1^\ell {(x_0^\ell)}^*}$.
For the left to right inclusion in \cref{eq:acc:Rconnection},
consider a marking $M\in\Coverset{\denotationof{S_\ell}}\setminus \denotationof{S_\ell}$.
We show that $M\in\Coverset{\denotationof{R}}$.
Recall that $\Coverset{\denotationof{S_\ell}}$ consists of all those markings $M$ so that there exists a finite path
$$M_0\NDStep{*}M'_0\NTStep{d_1}M_1\NDStep{*}M'_1\NTStep{d_2}M_2\dots M'_{k-1}\NDStep{*}M_k$$
alternating between timed and (sequences of) discrete transition steps,
with $M_0\in\denotationof{S_\ell}$, $M_k\ge M$ and all $d_i\le \max(\fractof{M'_i})$.

By our choice of $M$, there must be a first
expression in the sequence which is not a member of $\denotationof{S_\ell}$.
Since 
$\denotationof{\SAT{S_\ell}} = \denotationof{S_\ell}$,
we can assume an index $i>0$ so that
$M_i\notin \denotationof{S_\ell}$
but $M'_{i-1}\in \denotationof{S_\ell}$ that is,
the step that takes us out of $\denotationof{S_\ell}$ is a timed step.

Because $\denotationof{S_\ell}=\bigcup_{i\ge 2}\denotationof{S_i}$, 
it must hold that
$M'_{i-1} \in \denotationof{S_j} =\denotationof{x_j^j(x_{j-1}^j)^* x_1^j(x_0^j)^*}$
for some index $j\ge 2$.
We claim that it already holds that
\begin{equation}
    \label{eq:acc:almostR}
    M'_{i-1} \in \denotationof{x_j^j{(x_{j-1}^j)}^*x_1^j}.
\end{equation}
Suppose not. 
If $d_i<\max(\fractof{M'_{i-1}})$ then
$M_i\in\denotationof{\emptyset S_j}\subseteq \denotationof{S_j}$ by \cref{lem:epsilonsteps}, contradiction.
Otherwise, if $d_i=\max(\fractof{M'_{i-1}})$,
notice that every
abstraction $\abstr[S]{M'_{i-1}}\in\Lang{S_j}$
must have $\card{S}=4$.
So by \cref{lem:nonstarrotation},
$M_i\in \denotationof{(x_0^j+1)S_j}$.
But then again
\begin{equation}
\denotationof{(x_0^j+1)S_j}
\subseteq
\denotationof{\SAT{(x_0^j+1)S_j}}
\subseteq
\denotationof{S_{j+1}},
\end{equation}
contradicting our assumption
that $M_i\notin \denotationof{S_\ell}$.
Therefore \cref{eq:acc:almostR} holds. By \cref{lem:nonstarrotation} we derive that
$M_i\in\denotationof{(x_1^j+1)x_j^j(x_{j-1}^j)^*}
=\denotationof{(x_1^j+1)(x_{j-1}^j)^*}
\subseteq\denotationof{(x_1^\ell+1)(x_{\ell-1}^\ell)^*}
=\denotationof{R}$.
This concludes the proof of \cref{eq:acc:Rconnection}.

\medskip
\noindent
Notice that by \cref{lem:accstep} we have that 
\begin{equation}
\label{eq:acc:EisSunr1}
\Coverset{\denotationof{x_1 x_0^*}}
=\denotationof{\SAT{x_1 x_0^*}}
\cup\Coverset{\denotationof{\SAT{x_1 x_0^*}}}
=\denotationof{S_1}
\cup\Coverset{\denotationof{S_1}}.
\end{equation}
Analogously, we get
for every $i\ge 1$ that
\begin{align}
    \label{eq:acc:2}
\Coverset{\denotationof{E_i}}
= 
\denotationof{\SAT{E_i}}
\cup \Coverset{\denotationof{\SAT{(x^i_0+1)E_i}}}
=
  \denotationof{E_i}
  \cup \Coverset{\denotationof{E_{i+1}}}
\end{align}
This used \cref{lem:accstep} and the fact that $\SAT{E_i}=E_i$ by construction.
Using \cref{eq:acc:2}
and that $\denotationof{E_i}\subseteq\denotationof{E_{i+1}}$
for $i\ge 2$,
we deduce
\begin{equation}
    \label{eq:acc:3}
\Coverset{\denotationof{S_1}}
=\Coverset{\denotationof{E_1}}
=
\denotationof{E_1}
\cup
\left(
\bigcup_{i\ge 2}\Coverset{\denotationof{E_i}}
\right).
\end{equation}
Finally we can conclude the desired result as follows.
\begin{align*}
\Coverset{\denotationof{x_1 x_0^*}}
&\eqby{(\ref{eq:acc:EisSunr1})}
\denotationof{S_1}
\cup
\Coverset{\denotationof{S_1}}
\eqby{(\ref{eq:acc:3})}
\denotationof{S_1}
\cup
\Coverset{
\bigcup_{i\ge 2}\denotationof{E_i}
}\\
&\eqby{(\ref{eq:acc:EisS})}
\denotationof{S_1}
\cup
\Coverset{\denotationof{S_\ell}}\\
&\eqby{(\ref{eq:acc:Rconnection})}
\denotationof{S_1}
\cup
\denotationof{S_\ell}
\cup
\Coverset{\denotationof{R}}
\qedhere
\end{align*}
\end{proof}

\subsection{Main Result}
\label{sec:central}
The following theorem summarizes our main claims regarding the \ECOVER problem.

\begin{theorem}\label{thm:central}
Consider an instance of \ECOVER with 
${\cal N} = (P,T, \Var, G, \Pre,\Post)$ 
a non-consuming TPN where $\cmax$ is the largest constant appearing in the
transition guards $G$ encoded in unary, 
and let $p$ be an initial place and $t$ be a transition.
\begin{enumerate}
\item\label{thm:cent:1}
The number of different simple expressions of length $m$ is 
$B(m) \eqdef 2^{(\card{P} \cdot (\cmax+2) \cdot m)+m}$.
\item\label{thm:cent:2}
It is possible to compute a symbolic representation of the set of 
markings coverable from some marking in the initial set
$\N\cdot \{(p,{0})\}$, as a finite set of simple expressions.
I.e., one can compute simple expressions $S_1,\dots,S_\ell$ s.t.
$\bigcup_{1\le i\le \ell}\denotationof{S_i} =  \Coverset{\N\cdot \{(p,{0})\}}$
and where $\ell \le 3\cdot B(2)$.
Each of the $S_i$ has length either $2$ or $4$.
\item\label{thm:cent:3}
Checking if there exists $M\in\Coverset{\N\cdot \{(p,0)\}}$ with $M\dstep{t}$
can be done in $\?O(\card{P}\cdot\cmax)$ deterministic space.
\end{enumerate}
\end{theorem}
\begin{proof}
For \cref{thm:cent:1} note that a simple expression is described by
a word where some symbols have a Kleene star.
There are 
$\card{\Sigma}^m$ different words of length $m$
and $2^m$ possibilities to attach stars to symbols.
Since the alphabet is $\Sigma \eqdef 2^{P\x\zeroto{\cmax+1}}$
and $\card{\zeroto{\cmax+1}}=\cmax+2$, the result follows.

\smallskip
Towards \cref{thm:cent:2}, we can assume w.l.o.g.\ that our TPN is
non-consuming by \cref{lem:mtpn:wlog}, and thus the region abstraction introduced
in \cref{sec:regionabs} applies.
In particular, the initial set of markings $\N\cdot \{(p,{0})\}$
is represented exactly by the expression
$S_0 \eqdef\{(p,0)\}\emptyset^*$ where $\emptyset \in \Sigma$ is the symbol
corresponding to the empty set.
That is, we have $\denotationof{S_0} = \N\cdot \{(p,{0})\}$ and thus
$\Reachset{\denotationof{S_0}} = \Reachset{\N\cdot \{(p,{0})\}}$.

The claimed expressions $S_i$ are the result of iterating
\cref{alg:inner_loop} until a previously seen expression is revisited.
Starting at $i=0$ and $S_0 \eqdef\{(p,0)\}\emptyset^*$,
each round will set $S_{i+1},S_{i+2}$ and $S_{i+3}$ to the result of applying \cref{alg:inner_loop} to $S_i$,
and increment $i$ to $i+3$.

Notice that then all $S_i$ are simple expressions of length $2$ or $4$
and that in particular, all expressions with index divisible by $3$ 
are of the form $ab^*$ for $a,b\in\Sigma$.
Therefore after at most $B(2)$ iterations, an expression $S_\ell$ is revisited
(with $\ell\le 3B(2)$).
Finally, an induction using \cref{lem:acc:correctness} provides that
$\bigcup_{1\le i\le \ell}\denotationof{S_i} =  \Coverset{\N\cdot \{(p,{0})\}}$.

\smallskip
Towards \cref{thm:cent:3}, 
we modify the above algorithm for the \ECOVER problem with the sliding window technique.
The algorithm is the same as above where instead of recording all the expressions $S_1,\dots,S_\ell$,
we only store the most recent ones and uses them to decide whether the transition $t$ is enabled.
If the index $i$ reaches the maximal value of $3\cdot B(2)$ we return unsuccessfully.

The bounded index counter uses $\?O(\log(B(2)))$ space;
\Cref{alg:inner_loop} uses space $\?O(\log(B(5)))$ because it stores only simple expressions of length $\le 5$.
The space required to store the three expressions resulting from each application of \cref{alg:inner_loop} is $\?O(3 \cdot \log(B(4)))$.
For every encountered simple expression we can check in logarithmic space whether the transition $t$ is enabled by some marking in its denotation. 
Altogether the space used by our new algorithm is bounded by $\?O(\log(B(5)))$.
By \cref{thm:cent:1}, this is $\?O(|P|\cdot (\cmax+2))=\?O(\card{P}\cdot\cmax)$.
\end{proof}
\begin{corollary}
    The \ECOVER problem for TPN is \PSPACE-complete.
\end{corollary}
\begin{proof}
The \PSPACE\ lower bound was shown in \cref{thm:LB}.
The upper bound follows from \cref{lem:mtpn:wlog} and \cref{thm:cent:3} of \cref{thm:central}.
\end{proof}

\section{Conclusion and Future Work}
\label{sec:conclusion}
We have shown that \emph{Existential Coverability} (and its dual of universal
safety) is \PSPACE-complete for TPN with one real-valued clock per token.
This implies the same complexity for checking safety of arbitrarily large
timed networks without a central controller.
The absence of a central controller makes a big difference, since the
corresponding problem \emph{with} a central controller is 
complete for $F_{\omega^{\omega^\omega}}$ \cite{HSS2012}.

It remains an open question whether these positive results for the 
controller-less case can be generalized
to multiple real-valued clocks per token.
In the case \emph{with} a controller, safety becomes undecidable already for
two clocks per token \cite{ADM2004}.

Another question is whether our results can be extended to
more general versions of timed Petri nets.
In our version, clock values are either inherited, advanced as time passes, or reset to zero.
However, other versions of TPN allow the creation of output-tokens with 
new non-deterministically chosen non-zero clock values, 
e.g., the timed Petri nets of \cite{AMM07,AA2001}
and the read-arc timed Petri nets of \cite{BHR2006}.

\bibliography{main}

\newpage
\appendix
\label{sec:appendix}
\section{Proof of \cref{lem:discall}}
\begin{lemma}
\label{lem:discall-starfree}
For every non-consuming TPN $\?N$ there are polynomial time computable functions
$f: \Sigma \times \Sigma \times \Sigma \to \Sigma$
and 
$g: \Sigma \times \Sigma \times \Sigma \to \Sigma$ with the following
properties.
    \begin{enumerate}
    \item\label{lem:discall-starfree-ad1}
      $f$ and $g$ are monotone (w.r.t.\ subset ordering) in each argument.
    \item\label{lem:discall-starfree-ad2} 
      $f(\alpha, \beta, x) \supseteq x$ and $g(\alpha, \beta, x) \supseteq x$ for all $\alpha,\beta,x \in \Sigma$.
    \item\label{lem:discall-starfree-ad3}
      For every word $w=x_0x_1\dots x_k$ over $\Sigma$, $\alpha \eqdef x_0$ and
      $\beta \eqdef \bigcup_{i>0} x_i$, and \\
      $w' \eqdef f(\alpha,\beta,x_0)g(\alpha,\beta,x_1)\dots g(\alpha,\beta,x_k)$
      we have $\denotationof{w'} = 
    \{M'' \mid \exists M \in \denotationof{w} \land M\Dstep{*}M' \ge M''\}
    $.
\end{enumerate}
\end{lemma}
\begin{proof} (Sketch).
It suffices to show the existence of such functions $f_t$ and $g_t$ 
for individual transitions $t \in T$ and $\dstep{t}$ instead of
$\Dstep{*}$. The functions $f$ and $g$ can then be obtained by iterated
applications of $f_t$ and $g_t$ (for all transitions $t$) until convergence.
(In addition to expanding $x$, 
the results of each application $f_t$ and $g_t$ are also added to $\alpha$ and
$\beta$, respectively.)
This works, because the functions $f_t$ and $g_t$ are monotone and operate on the
finite domain/range $\Sigma$.
Since we have a polynomial number of transitions, and 
each symbol in $\Sigma$ can increase (by strict subset ordering)
at most $\card{P} \cdot (\cmax+1)$ times, the number of iterations 
is polynomial.
Moreover, the properties of \cref{lem:discall-starfree-ad1}, \cref{lem:discall-starfree-ad2}
and \cref{lem:discall-starfree-ad3} carry over directly from 
$f_t$ and $g_t$ to $f$ and $g$, respectively.

Now we consider the definitions and properties of the functions $f_t$ and
$g_t$ for a particular transition $t$.
Given a variable evaluation $\pi:\Var\to \nnreals$, we define the 
functions $\pi_0$ and $\pi_{>0}$ from sets over $(P\x\Var)$
to sets over $(P\x\N)$ as follows. Intuitively, they cover the parts of 
the assignment $\pi$ with zero/nonzero fractional values, respectively.
Let $\pi_0(S) \eqdef \{(p,c) \,|\, (p,y) \in S\ \wedge\ \pi(y)=c \in \N\}$
and
$\pi_{>0}(S) \eqdef \{(p,c) \,|\, (p,y) \in S\ \wedge\ \lfloor\pi(y)\rfloor=c
\ \wedge\ \fract(\pi(y)) >0\}$.
The definitions are lifted to multisets in the straightforward way. 

Now let $t$ be a transition. We say that $(\alpha,\beta)$ enables $t$
iff
$\exists \pi:\Var\to \nnreals$ such that 
$\pi(y)\in G(t)(y)$ for all variables $y$ and 
$\pi_0(\Pre(t)) \subseteq \alpha$
and $\pi_{>0}(\Pre(t)) \subseteq \beta$. 
Thus if $\abstr{M} = x_0x_1\dots x_n$ then $M$ enables $t$ iff
$(x_0, \bigcup_{i>0} x_i)$ enables $t$, since all transition guards in $G(t)$
are intervals bounded by integers (i.e., $t$ cannot distinguish 
between different nonzero fractional values). Moreover, enabledness can be checked in 
polynomial time (choose integers for the part in $\alpha$ and
rationals with fractional part $1/2$ for the part in $\beta$).

In the case where $(\alpha,\beta)$ does not enable $t$ we just let
$g_t(\alpha,\beta,x) \eqdef x$ and $f_t(\alpha,\beta,x) \eqdef x$.
The conditions above are trivially satisfied in this case.

In the case where $(\alpha,\beta)$ enables $t$, let
$g_t(\alpha,\beta,x) \eqdef x \cup \gamma$ where
$\gamma$ is defined as follows. We have $(p,c) \in \gamma$ iff
there is a $(p,y) \in \Post(t)$ and $(q, y) \in \Pre(t)$ such that
$(q,c) \in x$.
Similarly, let $f_t(\alpha,\beta,x) \eqdef x \cup \gamma$ where
$\gamma$ is defined as follows. We have $(p,c) \in \gamma$ iff
either 
(1)
there is a $(p,y) \in \Post(t)$ and $(q, y) \in \Pre(t)$ such that
$(q,c) \in x$, or
(2) $c=0$ and there is a $(p,0) \in \Post(t)$.
All these conditions can be checked in polynomial time.
\cref{lem:discall-starfree-ad1} and \cref{lem:discall-starfree-ad2} follow
directly from the definition.

Towards \cref{lem:discall-starfree-ad3}, we show
$\denotationof{w'} \supseteq \{M'' \mid \exists M \in \denotationof{w} \land M
\dstep{t} M' \ge M''\}$.
(The proof of the reverse inclusion $\subseteq$ is similar.)
Let 
$w=x_0x_1\dots x_k$, $\alpha \eqdef x_0$, $\beta \eqdef \bigcup_{i>0} x_i$ such that 
$(\alpha,\beta)$ enables $t$ and 
$w' \eqdef f_t(\alpha,\beta,x_0)g_t(\alpha,\beta,x_1)\dots g_t(\alpha,\beta,x_k)$.
If $M \in \denotationof{w}$ and $M \dstep{t} M'$ then $M' \ge M$ since $\?N$ is
non-consuming. We show that every additional token $(p,u) \in M' \ominus M$
is included in $\denotationof{w'}$. (This implies the inclusion above, since
$M' \ominus M \ge M'' \ominus M$.)
For every additional token $(p,u) \in M' \ominus M$ there are two
cases. 
\begin{itemize}
\item
Assume $\fract(u) >0$. Then the token $(p,u)$ must have inherited its
clock value from some token $(q,u) \in M$ via a variable $y$ specified in the 
Pre/Post of $t$ (since discrete transitions cannot create new fractional
parts of clock values).  This case is covered by $\gamma$ in the definition of $g_t$
above.
In particular, if $(q,u) \in M$ was abstracted to $x_i$ in $w$ then 
$(p,u) \in M'$ is
abstracted to $g_t(\alpha,\beta,x_i)$ in $w'$.
\item
Assume $\fract(u)=0$. Then there are two cases. In the first case the token $(p,u)$
inherited its clock value from some token $(q,u) \in M$ via a variable $y$ specified in the 
Pre/Post of $t$. This case is covered by part (1) of $\gamma$ in the definition of $f_t$
above. In particular, $(q,u) \in M$ was abstracted to $x_0$ in $w$, because 
$\fract(u)=0$. Thus $(p,u) \in M'$ is abstracted to $f_t(\alpha,\beta,x_0)$ in $w'$.
In the second case the token $(p,u)$ got its clock value via a clock-reset to zero.
This case is covered by part (2) of $\gamma$ in the definition of $f_t$
above. In particular, in this case we must have $u=0$, and $(p,0) \in M'$ was
abstracted to $f_t(\alpha,\beta,x_0)$ in $w'$.
\end{itemize}
It follows that $\abstr{M'} \le w'$, i.e., by the ordering on symbols in
$\Sigma$, every letter in $\abstr{M'}$ is smaller than the corresponding letter
in $w'$. 
Thus $M' \in \denotationof{w'}$. Since $M' \ge M''$ and
$\denotationof{w'}$ is downward closed, we also have 
$M'' \in \denotationof{w'}$ as required.
\qedhere
\end{proof}

\lemdiscall*
\begin{proof}
Let $f$ and $g$ be the functions from \cref{lem:discall-starfree}, which immediately
yields \cref{lem:discall-ad1} and \cref{lem:discall-ad2}.
Towards \cref{lem:discall-ad3}, consider all words $w$ in $\?L(E)$ 
that contain each starred symbol in $E$ at least once.
(The other cases are irrelevant for $\denotationof{E}$ since they are subsumed
by monotonicity.)
For each such word $w$, the $\alpha, \beta$ derived from $w$ in 
\cref{lem:discall-starfree} are the same as the $\alpha, \beta$ derived from
$E$ in \cref{lem:discall-ad3}.
If $x_i$ in $E$ carries a star then $w$ contains
a corresponding nonempty subsequence $x_i\dots x_i$. 
We apply \cref{lem:discall-starfree} to each such $w$ to obtain the corresponding $w'$.
The word $w'$ then contains the corresponding subsequence
$g(\alpha,\beta,x_i)\dots g(\alpha,\beta,x_i)$.
Let $E'$ then be defined as in \cref{lem:discall-ad3}, i.e., by applying
functions to the symbols and keeping the stars at the same
symbols as in $E$.
By \cref{lem:discall-starfree}, this is computable in polynomial time.
We have $\?L(E') = \bigcup_{w \in \?L(E)} \{w'\}$. 
Thus $\denotationof{E'} = \bigcup_{w \in \?L(E)} \denotationof{w'}
=
\bigcup_{w \in \?L(E)} \{M'' \mid \exists M \in \denotationof{w} \land
M\Dstep{*}M' \ge M''\}
=
\{M'' \mid \exists M \in \denotationof{E} \land M\Dstep{*}M' \ge M''\}
$ for \cref{lem:discall-ad3} as required.
\qedhere
\end{proof}
\end{document}